\documentclass[11pt,letterpaper]{article}   % use for most recent manuscripts

\usepackage{amsmath,amsfonts,amsthm,amssymb,stmaryrd}

\theoremstyle{plain}

\numberwithin{equation}{section}
\newtheorem{thm}{Theorem}[section]
\newtheorem{lem}[thm]{Lemma}
\newtheorem{cor}[thm]{Corollary}

\newcommand{\complex}{{\mathbb C}}
\newcommand{\positive}{{\mathbb N}}
\newcommand{\real}{{\mathbb R}}
\newcommand{\ascript}{{\mathcal A}}
\newcommand{\cscript}{{\mathcal C}}
\newcommand{\dscript}{{\mathcal D}}
\newcommand{\pscript}{{\mathcal P}}
\newcommand{\qscript}{{\mathcal Q}}
\newcommand{\rscript}{{\mathcal R}}
\newcommand{\sscript}{{\mathcal S}}
\newcommand{\tscript}{{\mathcal T}}
\newcommand{\rmtr}{\mathrm{tr}}
\newcommand{\rmcyl}{\mathrm{cyl}}
\newcommand{\rmim}{\mathop{Im}}
\newcommand{\ahat}{\widehat{a}}
\newcommand{\dhat}{\widehat{\dscript}}
\newcommand{\rhat}{\widehat{\rscript}}
\newcommand{\that}{\widehat{\tscript}}
\newcommand{\phat}{\widehat{p}}
\newcommand{\mutilde}{\widetilde{\mu}}

\newcommand{\ab}[1]{\left|#1\right|}
\newcommand{\doubleab}[1]{\left\|#1\right\|}
\newcommand{\brac}[1]{\left\{#1\right\}}
\newcommand{\paren}[1]{\left(#1\right)}
\newcommand{\sqbrac}[1]{\left[#1\right]}
\newcommand{\elbows}[1]{{\left\langle#1\right\rangle}}
\newcommand{\ket}[1]{{\left|#1\right>}}
\newcommand{\bra}[1]{{\left<#1\right|}}

\begin{document}

\title{A DYNAMICS FOR\\DISCRETE QUANTUM GRAVITY
}
\author{S. Gudder\\ Department of Mathematics\\
University of Denver\\ Denver, Colorado 80208, U.S.A.\\
sgudder@du.edu
}
\date{}
\maketitle

\begin{abstract}
This paper is based on the causal set approach to discrete quantum gravity. We first describe a classical sequential growth process (CSGP) in which the universe grows one element at a time in discrete steps. At each step the process has the form of a causal set (causet) and the ``completed'' universe is given by a path through a discretely growing chain of causets. We then quantize the CSGP by forming a Hilbert space $H$ on the set of paths. The quantum dynamics is governed by a sequence of positive operators $\rho _n$ on $H$ that satisfy normalization and consistency conditions. The pair $\paren{H,\brac{\rho _n}}$ is called a quantum sequential growth process (QSGP). We next discuss a concrete realization of a QSGP in terms of a natural quantum action. This gives an amplitude process related to the ``sum over histories'' approach to quantum mechanics. Finally, we briefly discuss a discrete form of Einstein's field equation and speculate how this may be employed to compare the present framework with classical general relativity theory.
\end{abstract}

\section{Introduction}  % Section 1
This paper builds on the causal set approach to discrete quantum gravity \cite{blms87, sor03} and we refer the reader to \cite{hen09, sur11} for more details and motivation. The origins of this approach stem from studies of the causal structure
$(M,<)$ of a Lorentzian space-time manifold $(M,g)$. For $a,b\in M$ we write $a<b$ if $b$ is in the causal future of $a$. If
$a\le b$ or $b\le a$ we say that $a$ and $b$ are \textit{comparable} and otherwise $a$ and $b$ are \textit{incomparable}. If there are no closed causal curves in $(M,g)$, then $(M,<)$ is a partially ordered set (poset). It has been shown that $(M,<)$ determines the topology and even the differential structure of the manifold $(M,g)$ \cite{sor03, sur11} and it is believed that the order structure $(M,<)$ provides a viable candidate for describing a discrete quantum gravity.

To remind us that we are dealing with a causal structure, we call a finite poset a \textit{causet}. A causet is assumed to be unlabeled and isomorphic causets are identified. In the literature, causets are frequently labeled according to the order of ``birth'' and this causes complications because covariant properties are independent of labeling \cite{hen09, sor03, sur11}. In this way, our causets are automatically covariant.

Section~2 describes a classical sequential growth process in which the universe grows one element at a time in discrete steps. At each step, the process has the form of a causet and the ``completed'' universe is given by a path through a discretely growing chain of causets. The transition probability at each step is given by an expression due to Rideout-Sorkin 
\cite{rs00, vr06}. Letting $\Omega$ be the set of paths, $\ascript$ be the $\sigma$-algebra generated by cylinder sets and
$\nu _c$ the probability measure determined by the transition probabilities, the classical dynamics is described by a Markov chain in the probability space $(\Omega ,\ascript ,\nu _c)$.

In Section~3 we quantize the classical framework by forming the Hilbert space $H=L_2(\Omega ,\ascript ,\nu _c)$. The quantum dynamics is governed by a sequence of positive operators $\rho _n$ on $H$ that satisfy normalization and consistency conditions. We employ $\rho _n$ to construct decoherence functionals and a quantum measure $\mu$ on a ``quadratic algebra'' $\sscript$ of subsets of $\Omega$. In general, $\sscript$ is between the collection of cylinder sets and
$\ascript$. We then nominate $(\Omega ,\sscript ,\mu )$ as a candidate model for a discrete quantum gravity.

Section~4 discusses a concrete realization of the quantum sequential growth process $\rho _n$ considered in Section~3. This realization is given in terms of a natural quantum action and is called an amplitude process. The amplitude process is related to the ``sum over histories'' approach to quantum mechanics \cite{hen09, sur11}. Section~5 briefly discusses a discrete form of Einstein's field equation and speculates how this may be employed to compare the present framework with classical general relativity theory.

\section{Classical Sequential Growth Processes} % Section 2
Let $\pscript _n$ be the collection of all causets of cardinality $n$, $n=1,2,\ldots$, and let $\pscript =\cup\pscript _n$ be the collection of all causets. An element $a\in x$ for $x\in\pscript$ is \textit{maximal} if there is no $b\in x$ with $a<b$. If
$x\in\pscript _n$, $y\in\pscript _{n+1}$, then $x$ \textit{produces} $y$ if $y$ is obtained from $x$ by adjoining a single maximal element $a$ to $x$. In this case we write $y=x\shortuparrow a$ and use the notation $x\to y$. If $x\to y$, we also say that
$x$ is a \textit{producer} of $y$ and $y$ is an \textit{offspring} of $x$. Of course, $x$ may produce many offspring and a causet may be the offspring of many producers. Moreover, $x$ may produce $y$ in various isomorphic ways. To describe this precisely, we consider labeled causets where a \textit{labeling} of $x\in\pscript _n$ is a function
$\ell\colon x\to\brac{1,2,\ldots ,n}$ such that $a,b\in x$ with $a<b$ implies $\ell (a)<\ell (b)$. In Figure~1, the labeled causet $x$ produces the labeled causets $u,v,w$. In this paper we mainly consider unlabeled causets (which we simply call causets) and identify isomorphic copies of a causet so we identify $u,v,w$ and say the \textit{multiplicity} of $x\to u$ is three and write
$m(x\to u)=3$. We then replace Figure~1 by the simpler Figure~2.

% Figure 1
\setlength{\unitlength}{8pt}
\begin{picture}(5,25)
\put(12,8){\circle{7}}   % circle x
\put(10.5,8){\circle*{.35}}
\put(10.25,7){\makebox{$\scriptstyle 1$}}
\put(12,8){\circle*{.35}}
\put(11.75,7){\makebox{$\scriptstyle 2$}}
\put(13.5,8){\circle*{.35}}
\put(13.25,7){\makebox{$\scriptstyle 3$}}
\put(12,4){\makebox{$x$}}
\put(12,10.5){\vector(0,1){4}}
\put(9.5,8){\vector(-1,1){7.5}}
\put(14.5,8){\vector(1,1){7.5}}
\put(-1.5,18){\makebox{$\scriptstyle 4$}}
\put(0,17){\circle{7}}   % circle u
\put(-1.5,16.2){\circle*{.35}}
\put(-1.7,15.2){\makebox{$\scriptstyle 1$}}
\put(0,16.2){\circle*{.35}}
\put(-.2,15.2){\makebox{$\scriptstyle 2$}}
\put(1.5,16.2){\circle*{.35}}
\put(1.2,15.2){\makebox{$\scriptstyle 3$}}
\put(-1.5,16.2){\line(1,2){.8}}
\put(0,16.2){\line(-1,2){.8}}
\put(-.75,17.85){\circle*{.35}}
\put(-4,16){\makebox{$u$}}
\put(12,17){\circle{7}}   % circle v
\put(10.5,16.2){\circle*{.35}}
\put(10.3,15.2){\makebox{$\scriptstyle 1$}}
\put(12,16.2){\circle*{.35}}
\put(11.75,15.2){\makebox{$\scriptstyle 2$}}
\put(13.5,16.2){\circle*{.35}}
\put(13.2,15.2){\makebox{$\scriptstyle 3$}}
\put(13.5,16.2){\line(-1,2){.8}}
\put(12,16.2){\line(1,2){.8}}
\put(12.75,17.8){\circle*{.35}}
\put(11.5,18){\makebox{$\scriptstyle 4$}}
\put(8,16){\makebox{$v$}}
\put(24,17){\circle{7}}   % circle w
\put(22.5,16.2){\circle*{.35}}
\put(22.5,15.2){\makebox{$\scriptstyle 1$}}
\put(24,16.2){\circle*{.35}}
\put(23.8,15.2){\makebox{$\scriptstyle 2$}}
\put(25.5,16.2){\circle*{.35}}
\put(25.2,15.2){\makebox{$\scriptstyle 3$}}
\put(22.5,16.2){\line(1,1){1.5}}
\put(25.5,16.2){\line(-1,1){1.5}}
\put(24,17.7){\circle*{.35}}
\put(23,18){\makebox{$\scriptstyle 4$}}
\put(20,16){\makebox{$w$}}
\end{picture}
\begin{picture}(25,12)
\put(32,8){\circle{7}}   % circle x
\put(30.5,8){\circle*{.35}}
\put(32,8){\circle*{.35}}
\put(33.5,8){\circle*{.35}}
\put(32,4){\makebox{$x$}}
\put(32,10.5){\vector(0,1){4}}
\put(32,17){\circle{7}}   % circle u
\put(30.5,16.2){\circle*{.35}}
\put(32,16.2){\circle*{.35}}
\put(33.5,16.2){\circle*{.35}}
\put(30.5,16.2){\line(1,2){.8}}%%
\put(32,16.2){\line(-1,2){.8}}
\put(31.25,17.8){\circle*{.35}}
\put(28,16){\makebox{$u$}}
\hskip 2pc{\textbf{Figure 1}\hskip 13pc\textbf{Figure 2}}
\end{picture}
\bigskip

The transitive closure of ${}\to$ makes $\pscript$ into a poset itself and we call $(\pscript ,\to)$ the \textit{causet growth process}. A  \textit{path} in $\pscript$ is a sequence (string) $\omega _1\omega _2\cdots$, where $\omega _i\in\pscript _i$ and
$\omega _i\to\omega _{i+1}$, $i=1,2,\ldots$. An $n$-path in $\pscript$ is a finite string $\omega _1\omega _2\cdots\omega _n$ where, again, $\omega _i\in\pscript$ and $\omega _i\to\omega _{i+1}$. We denote the set of paths by $\Omega$ and the set of $n$-paths by $\Omega _n$.

If $a,b\in x$ with $x\in\pscript$, we say that $a$ is an \textit{ancestor} of $b$ and $b$ is a \textit{successor} of $a$ if $a<b$. We say that $a$ is a \textit{parent} of $b$ and $b$ is a \textit{child} of $a$ if $a<b$ and there is no $c\in x$ with $a<c<b$. A
\textit{chain} in $x$ is a set of mutually comparable elements of $x$ and an \textit{antichain} in $x$ is a set of mutually incomparable elements of $x$. By convention, the empty set is both a chain and an antichain and clearly singleton sets also have this property. In Figure~2, any subset of $x$ is an antichain while $u$ has two chains of cardinality 2.

\begin{thm}       % Theorem 2.1
\label{thm21}
If $x$ is a labeled causet, then the number of labeled offspring of $x$ is the number of distinct antichains in $x$.
\end{thm}
\begin{proof}
Let $x$ have cardinality $n$ and suppose $A=\brac{a_1,\ldots ,a_k}$ is an antichain in $x$. Let $b\notin x$ and form the labeled causet $y=x\shortuparrow b$ where $y$ inherits the labeling of $x$, $\ell (b)=n+1$ and the elements of $A$ are the parents of $b$. Then $y$ is a labeled offspring of $x$ and different antichains give different labeled offspring because they give different parents of $b$. Conversely, if $y=x\shortuparrow b$ is a labeled offspring of $x$ and $A=\brac{a_1,\ldots ,a_k}$ is the set of parents of $b$, then $A$ is an antichain because $a_i< a_j$ contradicts the fact that $a_i$ is a parent of $b$. Hence, we have a bijection between the antichains in $x$ and labeled offspring of $x$.
\end{proof}

We denote the cardinality of a set $A$ by $\ab{A}$.

\begin{cor}       % Corollary 2.2
\label{cor22}
The number $r$ of labeled offspring of a labeled causet $x$ satisfies $\ab{x}+1\le r\le 2^{\ab{x}}$ and both bounds are achieved.
\end{cor}

\begin{cor}       % Corollary 2.3
\label{cor23}
The number of offspring of a causet $x$, including multiplicity, is the number of distinct antichains in $x$.
\end{cor}

\begin{cor}       % Corollary 2.4
\label{cor24}
The number $r$ of offspring of a causet $x$, including multiplicity, satisfies $\ab{x}+1\le r\le 2^{\ab{x}}$ and both bounds are achieved.
\end{cor}

For example, the causet $x$ in Figure~2 has eight distinct antichains so $x$ has eight offspring, including multiplicity. The causet $u$ in Figure~2 has ten distinct antichains so $u$ has ten offspring, including multiplicity.

Let $c=(c_0,c_1,\ldots )$ be a sequence of nonnegative numbers (called \textit{coupling constants}) \cite{sor94, vr06}. For
$r,s\in\positive$ with $r\le s$ define
\begin{equation*}
\lambda _c(s,r)=\sum _{k=r}^s\binom{s-r}{k-r}c_k=\sum _{k=0}^{s-r}\binom{s-r}{k}c_{r+k}
\end{equation*}
For $x\in\pscript _n$, $y\in\pscript _{n+1}$ with $y=x\shortuparrow a$ define the transition probability \cite{sor94, vr06}
\begin{equation*}
p_c(x\to y)=m(x\to y)\,\frac{\lambda _c(\alpha ,\pi )}{\lambda _c(n,0)}
\end{equation*}
where $\alpha$ is the number of ancestors and $\pi$ the number of parents of $a$. By convention we define $p_c(x\to y)=0$ if
$x\not\to y$. It is shown in \cite{vr06} that $p_c(x\to y)$ is a probability distribution in that it satisfies the Markov-sum rule
$\sum p_c(x\to y)=1$. The distribution $p_c(x\to y)$ is essentially the most general that is consistent with principles of causality and covariance \cite{rs00, vr06}. It is hoped that other theoretical principles or experimental data will determine the coupling constants One suggestion is to take $c_k=1/k!$ \cite{sor03}. Another is the \textit{percolation dynamics} $c_k=c^k$ for some $c>0$ \cite{sur11}. For this later choice we have the simple form
\begin{equation*}
p_c(x\to y)=m(x\to y)r^\pi(1-r)^\beta
\end{equation*}
where $\beta$ is the number of incomparable elements to $a$ (other than $a$ itself) and $r=c(1+c)^{-1}$ \cite{gud112}.

We view a causet $x\in\pscript _n$ as a possible universe at step $n$ while a path may be viewed as a possible (evolved) universe. The set $\pscript$ together with the set of transition probabilities $p_c(x\to y)$ forms a \textit{classical sequential growth process} (CSGP) which we denote by $(\pscript ,p_c)$ \cite{gud111, rs00, vr06}. It is clear that $(\pscript ,p_c)$ is a Markov chain. As with any Markov chain, the probability of an $n$-path $\omega =\omega _1\omega _2\cdots\omega _n$ is
\begin{equation*}
p_c^n(\omega )=p_c(\omega _1\to\omega _2)p_c(\omega _2\to\omega _3)\cdots p_c(\omega _{n-1}\to\omega _n)
\end{equation*}
Of course, $\omega\mapsto p_c^n(\omega )$ is a probability measure on $\Omega _n$. Figure~3 illustrates the first two steps of a CSGP. It follows from Corollary~\ref{cor23} that the number of offspring including multiplicity of $x_4$, $x_5$, $x_6$, $x_7$ and $x_8$ are $4,5,6,5,8$, respectively. In this case, Corollary~\ref{cor24} tells us that $4\le r\le 8$.
\vskip 2pc

% Figure 3
\setlength{\unitlength}{8pt}
\begin{picture}(45,26)
\put(22,8){\circle{7}}  % circle - x_1
\put(22,8){\circle*{.35}}
\put(17.5,7){\makebox{$x_1$}}
\put(19.5,8){\vector(-1,1){4}} % vector between x_1 and x_2
\put(24.5,8){\vector(1,1){4}} % vector between x_1 and x_3
\put(14,14){\circle{7}} % circle - x_2
\put(12,15.5){\vector(-1,1){5.5}} % vector between x_2 and x_4
\put(16,15,5){\vector(1,1){5}} % vector between x_2 and x_6
\put(14,15){\circle*{.35}}
\put(14,13){\circle*{.35}}
\put(14,13.2){\line(0,1){1.5}}
\put(9.75,13.75){\makebox{$x_2$}}
\put(14,16.5){\vector(0,1){4}} % vector between x_2 and x_5
\put(30,14){\circle{7}} % circle - x_3
\put(28,15.5){\vector(-1,1){5}} % vector between x_3 and x_6
\put(32,15.5){\vector(1,1){5}} % vector between x_3 and x_8
\put(25,17){\makebox{$\scriptstyle 2$}}
\put(29,14){\circle*{.35}}
\put(31,14){\circle*{.35}}
\put(33,13.5){\makebox{$x_3$}}
\put(30,16.5){\vector(0,1){4}} % vector between x_3 and x_7
\put(5,23){\circle{7}} % circle - x_4
\put(5,21.5){\circle*{.35}}
\put(5,23){\circle*{.35}}
\put(5,24.5){\circle*{.35}}
\put(5,21.5){\line(0,1){3}}
\put(.75,23){\makebox{$x_4$}}
\put(14,23){\circle{7}} % circle - x_5
\put(14,22){\circle*{.35}}
\put(13.25,23.5){\circle*{.35}}
\put(14.75,23.5){\circle*{.35}}
\put(14,22){\line(1,2){.8}}
\put(14,22){\line(-1,2){.8}}
\put(10,23){\makebox{$x_5$}}
\put(22,23){\circle{7}} % circle - x_6
\put(21,22){\circle*{.35}}
\put(22.5,22){\circle*{.35}}
\put(21,23.5){\circle*{.35}}
\put(21,22){\line(0,1){1.5}}
\put(18,23){\makebox{$x_6$}}
\put(30,23){\circle{7}} % circle - x_7
\put(29,22){\circle*{.35}}
\put(31,22){\circle*{.35}}
\put(30,23.75){\circle*{.35}}
\put(29,22){\line(1,2){.8}}
\put(31,22){\line(-1,2){.8}}
\put(26,23){\makebox{$x_7$}}
\put(38,23){\circle{7}} % circle - x_8
\put(36.5,23){\circle*{.35}}
\put(38,23){\circle*{.35}}
\put(39.5,23){\circle*{.35}}
\put(41,23){\makebox{$x_8$}}
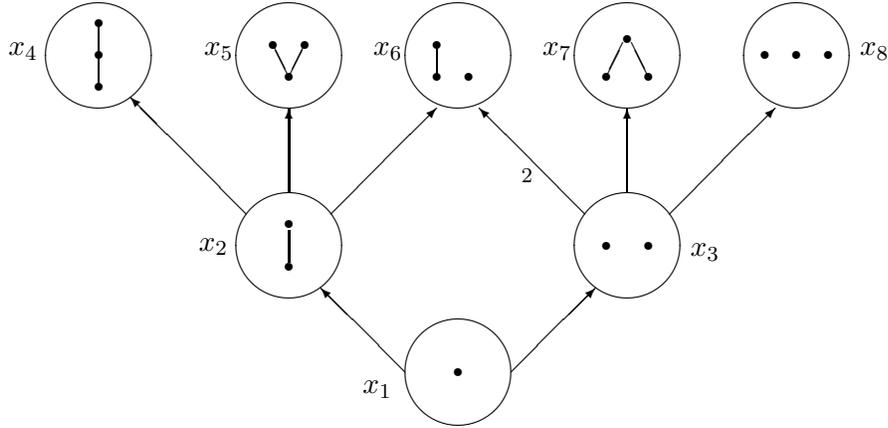
\centerline{\textbf{Figure 3}}
\end{picture}
\bigskip

The set of all paths beginning with $\omega =\omega _1\omega _2\cdots\omega _n\in\Omega _n$ is called an
\textit{elementary cylinder set} and is denoted by $\rmcyl (\omega )$. If $A\subseteq\Omega _n$, then the
\textit{cylinder set} $\rmcyl (A)$ is defined by
\begin{equation*}
\rmcyl (A)=\bigcup _{\omega\in A}\rmcyl (\omega )
\end{equation*}
Using the notation
\begin{equation*}
\cscript (\Omega _n)=\brac{\rmcyl (A)\colon A\subseteq\Omega _n}
\end{equation*}
we see that
\begin{equation*}
\cscript (\Omega _1)\subseteq\cscript (\Omega _2)\subseteq\cdots
\end{equation*}
is an increasing sequence of subalgebras of the \textit{cylinder algebra} $\cscript (\Omega )=\cup\cscript (\Omega _n)$. For
$A\subseteq\Omega$ we define the set $A^n\subseteq\Omega _n$ by
\begin{equation*}
A^n=\brac{\omega _1\omega _2\cdots\omega _n\in\Omega _n
   \colon\omega _1\omega _2\cdots\omega _n\omega _{n+1}\cdots\in A}
\end{equation*}
That is, $A^n$ is the set of $n$-paths that can be continued to a path in $A$. We think of $A^n$ as the $n$-step approximation to $A$.

For $A=\rmcyl (A_1)\in\cscript (\Omega _n)$, $A_1\subseteq\Omega _n$ define $\phat _c(A)=p_c^n(A)$. Notice that $\phat _c$ becomes a well-defined probability measure on the algebra $\cscript (\Omega )$. By the Kolmogorov extension theorem,
$\phat _c$ has a unique extension to a probability measure $\nu _c$ on the $\sigma$-algebra $\ascript$ generated by
$\cscript (\Omega )$. We conclude that $(\Omega ,\ascript ,\nu _c)$ is a probability space and the restriction
$\nu _c\mid\cscript (\Omega _n)=\phat _c$.

\section{Quantum Sequential Growth Processes} % Section 3
We now show how to ``quantize'' the CSGP $(\pscript ,p_c)$ to obtain a quantum sequential growth process (QSGP). Letting
$H=L_2(\Omega ,\ascript ,\nu _c)$ be the \textit{path Hilbert space} we see that the $n$-\textit{path Hilbert spaces}
$H_n=L_2(\Omega ,\cscript (\Omega _n),\phat _c)$ form an increasing sequence $H_1\subseteq H_2\subseteq\cdots$ of closed subspaces of $H$. Let
\begin{equation*}
\Omega '_n=\brac{\omega\in\Omega _n\colon p_c^n(\omega )\ne 0}
\end{equation*}
be the set of $n$-paths with nonzero measure. For $\omega\in\Omega '_n$, letting
\begin{equation*}
e_\omega ^n=\chi _{\rmcyl (\omega )}/p_c^n(\omega )^{1/2}
\end{equation*}
it is clear that $\brac{e_\omega ^n\colon\omega\in\Omega '_n}$ forms an orthonormal basis for $H_n$.

Letting $1=\chi _\Omega$ we see that 1 is a unit vector in $H$. More generally, if $A\in\cscript (\Omega )$ then the characteristic function $\chi _s\in H$ with $\doubleab{\chi _s}=\nu _c(A)^{1/2}$. If $T$ is an operator on $H_n$, we shall assume that $T$ is also an operator on $H$ by defining $Tf=0$ for all $f\in H_n^\perp$. A \textit{probability operator} on $H_n$ is a positive operator $\rho$ on $H_n$ that satisfies the normalization condition $\elbows{\rho 1,1}=1$. If $\rho$ is a probability operator on $H_n$, we define the \textit{decoherence functional}
$D_\rho\colon\cscript (\Omega _n)\times\cscript (\Omega _n)\to\complex$ by
\begin{equation*}
D_\rho (A,B)=\rmtr\paren{\rho\ket{\chi _B}\bra{\chi _A}}=\elbows{\rho\chi _B,\chi _A}
\end{equation*}
It is easy to check that $D_\rho$ has the usual properties of a decoherence functional. Namely, $D_\rho (\Omega ,\Omega )=1$, $D_\rho (A,B)=\overline{D_\rho (B,A)}$, $A\mapsto D_\rho (A,B)$ is a complex measure on $\cscript (\Omega _n)$ and if
$A_i\in\cscript (\Omega _n)$, $i=1,\ldots ,r$, then the $r\times r$ matrix with components $D_\rho (A_i,A_j)$ is positive semidefinite. We interpret $D_\rho (A,B)$ as a measure of the interference between the events $A$ and $B$ when the system is described by $\rho$. We also define the $q$-\textit{measure} $\mu _\rho\colon\cscript (\Omega _n)\to\real ^+$ by
$\mu _\rho (A)=D_\rho (A,A)$ and interpret $\mu _\rho (A)$ as the quantum propensity of the event $A$. In general, $\mu _\rho$ is not additive so $\mu _\rho$ is not a measure on $\cscript (\Omega _n)$. However, $\mu _\rho$ is \textit{grade}-2
\textit{additive} \cite{djs10, sor94, sor03} in the sense that if $A,B,C\in\cscript (\Omega _n)$ are mutually disjoint, then
\begin{align}         % equation (3.1)
\label{eq31}
\mu _\rho&(A\cup B\cup C)\notag\\
   &=\mu _\rho (A\cup B)+\mu _\rho (A\cup C)+\mu _\rho (B\cup C)-\mu _\rho (A)-\mu _\rho (B)-\mu _\rho (C)
\end{align}

A subset $\qscript\subseteq\ascript$ is a \textit{quadratic algebra} if $\emptyset,\Omega\in\qscript$ and if $A,B,C\in\qscript$ are mutually disjoint with $A\cup B,A\cup C,B\cup C\in\qscript$, then $A\cup B\cup C\in\qscript$. A $q$-\textit{measure} on a quadratic algebra $\qscript$ is a map $\mu\colon\qscript\to\real ^+$ satisfying \eqref{eq31} whenever, $A,B,C\in\qscript$ are mutually disjoint with $A\cup B,A\cup C,B\cup C\in\qscript$. In particular $\cscript (\Omega _n)$ is a quadratic algebra and
$\mu _\rho\colon\cscript (\Omega _n)\to\real ^+$ is a $q$-measure in this sense.

Let $\rho _n$ be a probability operator on $H_n$, $n=1,2,\ldots$, which we view as a probability operator on $H$. We say that the sequence $\rho _n$ is \textit{consistent} if
\begin{equation*}
D_{\rho _{n+1}}(A,B)=D_{\rho _n}(A,B)
\end{equation*}
for every $A,B\in\cscript (\Omega _n)$. We call a consistent sequence $\rho _n$ a \textit{discrete quantum process} and we call
$\paren{H,\brac{\rho _n}}$ a \textit{quantum sequential growth process} (QSGP).

Let $\paren{H,\brac{\rho _n}}$ be a QSGP. If $C\in\cscript (\Omega )$ has the form $C=\rmcyl (A)$, $A\in\cscript (\Omega _n)$, we define $\mu (C)=\mu _{\rho _n}(A)$. It is easy to check that $\mu$ is well-defined and gives a $q$-measure on
$\cscript (\Omega )$. In general, $\mu$ cannot be extended to a $q$-measure on $\ascript$, but it is important to extend $\mu$ to other physically relevant sets \cite{djs10, sor11}. We say that a set $A\in\ascript$ is \textit{suitable} if
$\lim\rmtr\paren{\rho _n\ket{\chi _A}\bra{\chi _A}}$ exists and is finite and if this is the case, we define $\mutilde (A)$ to be the limit. We denote the collection of suitable sets by $\sscript (\Omega )$ and it is shown in \cite{gud112} that $\sscript (\Omega )$ is a quadratic algebra with $\mutilde$ a $q$-measure on $\sscript (\Omega )$ that extends $\mu$ from $\cscript (\Omega )$.

In general, $\sscript (\Omega )$ is strictly between $\cscript (\Omega )$ and $\ascript$. For example, if $A\in\ascript$ with
$\nu _c(A)=0$, then $\chi _A=0$ almost everywhere so $\ket{\chi _A}\bra{\chi _A}=0$ and $\mutilde (A)=0$. To be specific, if
$\omega\in\Omega$ then $\brac{\omega}\in\ascript$ but $\brac{\omega}\notin\cscript (\Omega )$. Although there are exceptions, a typical $\omega\in\Omega$ satisfies $\nu _c\paren{\brac{\omega}}=0$ so $\mutilde\paren{\brac{\omega}}=0$. It follows from Schwarz's inequality that if $A\in\sscript (\Omega )$ then $\mutilde (A)\le\nu _c(A)\sup\doubleab{\rho _n}$.

\begin{thm}       % Theorem 3.1
\label{thm31}
If $\paren{H,\brac{\rho _n}}$ is a QSGP and $A\in\ascript$, then $A\in\sscript (\Omega )$ if and only if
$\lim\mu _{\rho _n}\sqbrac{\rmcyl (A^n)}$ exists. If this is the case, then
$\mutilde (A)\!=\!\lim\mu _{\rho _n}\!\sqbrac{\rmcyl (A^n)}$.
\end{thm}
\begin{proof}
Let $P_n$ be the orthogonal projection from $H$ onto $H_n$. We first show that $P_n\chi _A=\chi _{\rmcyl (A^n)}$. Now for $e_\omega ^n$, $\omega\in\Omega '_n$ with $\omega =\omega _1\omega _2\cdots\omega _n$ we have
\begin{align*}
\elbows{\chi _{\rmcyl (A^n)},e_\omega ^n}
  &=\frac{1}{p_c^n(\omega )^{1/2}}\elbows{\chi _{\rmcyl (A^n)},\chi _{\rmcyl (\omega )}}\\\noalign{\medskip}
  &=\begin{cases}{p_c^n(\omega )^{1/2}}&{\hbox{if  }\omega _1\omega _2\cdots\omega _n\omega _{n+1}\cdots\in A}\\
  {0}&{\hbox{otherwise}}\end{cases}\\\noalign{\medskip}
  &=\frac{1}{p_c^n(\omega )^{1/2}}\elbows{\chi _a,\chi _{\rmcyl (\omega )}}\\\noalign{\medskip}
  &=\elbows{\chi _A,e _\omega ^n}=\elbows{P_n\chi _A,e_\omega ^n}
\end{align*}
Since $P_n\chi _A,\chi _{\rmcyl (A^n)}\in H_n$ is follows that $P_n\chi _A=\chi _{\rmcyl (A^n))}$. Hence,
\begin{align*}
\rmtr\paren{\rho _n\ket{\chi _A}\bra{\chi _A}}&=\rmtr\paren{\rho _nP_n\ket{\chi _A}\bra{\chi _A}}\\
  &=\sum _{\omega\in\Omega '_n}\elbows{\rho _nP_n\ket{\chi _A}\bra{\chi _A}e_\omega ^n,e_\omega ^n}\\
  &=\sum _{\omega\in\Omega '_n}\elbows{e_\omega ^n,\chi _A}\elbows{\rho _nP_n\chi _{A^i}e_\omega ^n}\\
  &=\sum _{\omega\in\Omega '_n}\elbows{e_\omega ^n,\chi _{\rmcyl (A^n)}}\elbows{\rho _n\chi _{\rmcyl (A^n)},e_\omega ^n}\\
  &=\rmtr\paren{\rho _n\ket{\chi _{\rmcyl (A^n)}}\bra{\chi _{\rmcyl (A^n)}}}=\mu _{\rho _n}(A^n)
\end{align*}
The result now follows.
\end{proof}

\section{Amplitude Processes} % Section 4
Various ways of constructing discrete quantum processes on $H\!=\!L_2(\Omega ,\ascript ,\nu _c)$ have been considered
\cite{gud111, gud112, gud121}. After we introduce general amplitude processes, we present a concrete realization of a discrete quantum process in terms of a natural quantum action.

A \textit{transition amplitude} is a map $a\colon\pscript _n\times\pscript _{n+1}\to\complex$ such that $a(x,y)=a(x\to y)=0$ if
$p_c(x\to y)=0$ and
\begin{equation}         % equation (4.1)
\label{eq41}
\sum\brac{a(x\to y)\colon y\in\pscript _{n+1},x\to y}=1
\end{equation}
for all $x\in\pscript _n$. The \textit{amplitude process} (AP) corresponding to $a$ is given by the maps
$a_n\colon\Omega _n\to\complex$ where
\begin{equation*}
a_n(\omega _1\omega _2\cdots\omega _n)
  =a(\omega _1\to\omega _2)a(\omega _2\to\omega _3)\cdots a(\omega _{n-1}\to\omega _n)
\end{equation*}
We define the \textit{probability vector} $\ahat _n\colon\Omega _n\to\complex$ by $\ahat _n(\omega )=0$ if $p_c^n(\omega )=0$ and if $\omega\in\Omega '_n$ then $\ahat _n(\omega )=p_c^n(\omega )^{-1}a_n(\omega )$. For a given AP $a_n$ define the positive operators $\rho _n$ on $H_n$ by
\begin{equation*}
\elbows{\rho _n\chi _{\brac{\omega '}},\chi _{\brac{\omega}}}=a_n(\omega )\overline{a_n(\omega ')}
\end{equation*}
for all $\omega ,\omega '\in\Omega '_n$. Then
\begin{equation*}
\elbows{\rho _ne _{\omega '}^n,e_\omega ^n}
  =p_c^n(\omega ')^{-1/2}p_c^n(\omega )^{-1/2}a_n(\omega )\overline{a_n(\omega ')}
\end{equation*}
It follows that $\rho _n$ is the rank 1 operator given by $\rho _n=\ket{\ahat _n}\bra{\ahat _n}$.

\begin{thm}       % Theorem 4.1
\label{thm41}
The operators $\rho _n$, $n=1,2,\ldots$, form a discrete quantum process.
\end{thm}
\begin{proof}
We have seen that $\rho _n$ is a positive operator on $H_n$, $n=1,2,\ldots\,$. To show that $\rho _n$ is a probability operator we have
\begin{align}         % equation (4.2)
\label{eq42}
\elbows{e_n1,1}&=\elbows{\rho _n\sum\chi _{\brac{\omega}},\sum\chi _{\brac{\omega}}}
  =\sum _{\omega ,\omega '}\elbows{e_n\chi _{\brac{\omega}},\chi _{\brac{\omega '}}}\notag\\
  &=\sum _{\omega ,\omega '}a_n(\omega)\overline{a_n(\omega )}=\ab{\sum _\omega a_n(\omega )}^2
\end{align}
Applying \eqref{eq41} we obtain
\begin{align}         % equation (4.3)
\label{eq43}
\sum a_n(\omega )
  &=\sum a(\omega _1\to\omega _2)a(\omega _2\to\omega _3)\cdots a(\omega _{n-1}\to\omega _n)\notag\\
  &=\sum a(\omega _1\to\omega _2)\cdots a(\omega _{n-2}\to\omega _{n-1})
  \sum _{\omega _n}a(\omega _{n-1}\to\omega _n)\notag\\
  &=\sum a(\omega _1\to\omega _2)\cdots a(\omega _{n-2}\to\omega _{n-1})\notag\\
  &\quad\vdots\notag\\
  &=\sum _{\omega _2}a(\omega _1\to\omega _2)=1
\end{align}
By \eqref{eq42} and \eqref{eq43} we conclude that $\elbows{\rho _n1,1}=1$ so $P_n$ is a probability operator. To show that
$\rho _n$ is a consistent sequence, let $\omega ,\omega '\in\Omega '_n$ with
$\omega =\omega _1\omega _2\cdots\omega _n$, $\omega '=\omega '_1\omega '_2\cdots\omega '_n$. By \eqref{eq41} we have
\begin{align}         % equation (4.4)
\label{eq44}
D_{n+1}&\sqbrac{\rmcyl (\omega ),\rmcyl (\omega ')}
  =\elbows{\rho _{n+1}\chi _{\rmcyl (\omega ')},\chi _{\rmcyl (\omega )}}\notag\\
  &=\sum\brac{a_n(\omega )a_n(\omega _n\to x)\overline{a(\omega ')}a_n(\omega '_n\to y)
  \colon\omega _n\to x,\omega '_n\to y}\notag\\
  &=a_n(\omega )\overline{a_n(\omega ')}\sum\brac{a(\omega _n\to x
  \colon\omega _n\to x}\sum\brac{\overline{a(\omega '_n\to y)}\colon\omega '_n\to y}\notag\\
  &=a_n(\omega )\overline{a_n(\omega ')}=D_n\sqbrac{\rmcyl (\omega ),\rmcyl (\omega ')}
\end{align}
For $A,B\in\cscript (\Omega _n)$ by \eqref{eq44} we have
\begin{align*}
D_{n+1}(A,B)&=\sum\brac{D_{n+1}\sqbrac{\rmcyl (\omega ),\rmcyl (\omega ')}\colon\omega\in A,\omega '\in B}\\
  &=\sum\brac{D_n\sqbrac{\rmcyl (\omega ),\rmcyl (\omega ')}\colon\omega\in A,\omega '\in B}\\
  &=D_n(A,B)\qedhere
\end{align*}
\end{proof}

Since $\rho _n=\ket{\ahat _n}\bra{\ahat _n}$, we see that
\begin{equation*}
\doubleab{\rho _n}=\doubleab{\ket{\ahat _n}\bra{\ahat _n}}=\doubleab{\ahat _n}^2=\sum\ab{a_n(\omega )}^2
\end{equation*}
The decoherence functional corresponding to $\rho _n$ becomes
\begin{align*}
D_n(A,B)&=\elbows{\rho _n\chi _B,\chi _A}=\elbows{\ket{\ahat _n}\bra{\ahat _n}\chi _B,\chi _A}\\
  &=\elbows{\ahat _n,\chi _A}\elbows{\chi _B,\ahat _n}\\
  &=\sum\brac{a_n(\omega )\colon\omega\in A\cap\Omega '_n}
  \sum\brac{\overline{a_n(\omega ')}\colon\omega '\in B\cap\Omega '_n}
\end{align*}
The corresponding $q$-measure is given by
\begin{equation*}
\mu _n(A)=D_n(A,A)=\ab{\sum\brac{a_n(\omega )\colon\omega\in A\cap\Omega '_n}}^2
\end{equation*}
It follows from Theorem~\ref{thm31} that if $A\in\ascript$ then $A\in\sscript (\Omega )$ if and only if
\begin{equation*}
\lim\mu _{\rho _n}\sqbrac{\rmcyl (A^n)}=\lim\ab{\sum\brac{a_n(\omega )\colon\omega\in A_n\cap\Omega '_n}}^2
\end{equation*}
exists and is finite in which case $\mutilde (A)$ is this limit.

We now present a specific example of an AP that arises from a natural quantum action. For $x\in\pscript$, the \textit{height}
$h(x)$ of $x$ is the cardinality of a largest chain is $x$. The \textit{width} $w(x)$ of $x$ is the cardinality of a largest antichain in $x$. Finally, the \textit{area} $A(x)$ of $x$ is given by $A(x)=h(x)w(x)$. Roughly speaking, $h(x)$ corresponds to an internal time in $x$, $w(x)$ corresponds to the mass or energy of $x$ \cite{gud121} and $A(x)$ corresponds to an action for $x$. If
$x\to y$, then $h(y)=h(x)$ or $h(x)+1$ and $w(y)=w(x)$ or $w(x)+1$. In the case of $h(y)=h(x)+1$ we call $y$ a
\textit{height offspring} of $x$, in the case $w(y)=w(x)+1$ we call $y$ a \textit{width offspring} of $x$ and if both $h(y)=h(x)$,
$w(y)=w(x)$ hold we call $y$ a \textit{mild offspring} of $x$. Let $H(x)$, $W(x)$ and $M(x)$ be the sets of height, width and mild offspring of $x$, respectively.

\begin{lem}       % Lemma 4.2
\label{lem42}
The sets $H(x)$, $W(x)$, $M(x)$ form a partition of the set of offspring $x$.
\end{lem}
\begin{proof}
Since $\brac{y\colon x\to y}=H(x)\cup W(x)\cup M(x)$, we only need to show that $H(x)$, $W(x)$, $M(x)$ are mutually disjoint. Clearly, $H(x)\cap M(x)=W(x)\cap M(x)=\emptyset$ so we must show that $H(x)\cap W(x)=\emptyset$. Suppose $y\in H(x)$ where $y=x\shortuparrow a$. If $y\in W(x)$ then $a$ is incomparable with every element of some largest antichain
$\brac{b_1,\ldots ,b_r}$ in $x$. Also, $a>a_s>a_{s-1}>\cdots>a_1$ where $\brac{a_1,\ldots ,a_s}$ is a largest chain in $x$. It follows that $b_i\ne a_j$ for every $i,j$. Now $a_s\not> b_i$ for some $i$ because otherwise $a>b_i$ which is a contradiction. Hence, $a_s<b_i$ for some $i$ because otherwise $\brac{b_1,\ldots ,b_r,a_s}$ would be a larger antichain in $x$. But then
$\brac{a_1,\ldots ,a_s,b_i}$ is a larger chain in $x$. But this is a contradiction.
\end{proof}

If $x\to y$ we have the following three possibilities: $y\in M(x)$ in which case $A(y)-A(x)=0$, $y\in H(x)$ in which case
\begin{equation*}
A(y)-A(x)=\sqbrac{h(x)+1}w(x)-h(x)w(x)=w(x)
\end{equation*}
$y\in W(x)$ in which case
\begin{equation*}
A(y)-A(x)=h(x)\sqbrac{w(x)+1}-h(x)w(x)=h(x)
\end{equation*}
We define the transition amplitude $a(x\to y)$ in terms of the ``action'' change from $x$ to $y$. We first define the
\textit{partition function}
\begin{equation*}
z(x)=\sum _y\brac{e^{2\pi i\sqbrac{A(y)-A(x)}/\ab{x}}\colon p_c^{\ab{x}}(x\to y)\ne 0}
\end{equation*}
For $x\to y$ define the transition amplitude $a(x\to y)$ to be $0$ if $p_c^{\ab{x}}(x\to y)=0$ and otherwise
\begin{equation*}
a(x\to y)=\tfrac{1}{z(x)}e^{2\pi i\sqbrac{A(y)-A(x)}/\ab{x}}
\end{equation*}
As before, we have three possibilities. If $y\in M(x)$ then $a(x\to y)=z(x)^{-1}$, if $y\in H(x)$ then
\begin{equation*}
a(x\to y)=\frac{e^{2\pi iw(x)/\ab{x}}}{z(x)}
\end{equation*}
if $y\in W(x)$ then
\begin{equation*}
a(x\to y)=\frac{e^{2\pi ih(x)/\ab{x}}}{z(x)}
\end{equation*}

Since the transition amplitudes $a(x\to y)$ satisfy \eqref{eq41} it follows from Theorem~\ref{thm41} that the corresponding $\rho _n$ form a discrete quantum process. For any $x\in\pscript$ there are only three possible values for $a(x\to y)$. This is roughly analogous to a 3-dimensional Markov chain. Does this indicate the emergence of 3-dimensional space?

\section{Discrete Einstein Equation} % Section 5
Let $Q_n=\bigcup _{i=1}^n\pscript _i$ and let $K_n$ be the Hilbert space $\complex ^{Q_n}$ with the standard inner product
\begin{equation*}
\elbows{f,g}=\sum _{x\in Q_n}\overline{f(x)}g(x)
\end{equation*}
Let $L_n=K_n\otimes K_n$ which we identify with $\complex ^{Q_n\times Q_n}$. Suppose $\paren{H,\brac{\rho _n}}$ is a QSGP with corresponding decoherence matrices
\begin{equation*}
D_n(\omega ,\omega ')=D_n\sqbrac{\rmcyl (\omega ),\rmcyl (\omega ')}\quad\omega ,\omega '\in\Omega _n
\end{equation*}
If $\omega =\omega _1\omega _2\cdots\omega _n\in\Omega _n$ and $\omega _i=x$ for some $i$, then $\omega$ \textit{contains} $x$. For $x,y\in Q_n$ we define
\begin{equation*}
D_n(x,y)=\sum\brac{D_n(\omega ,\omega ')\colon\omega\hbox{ contains }x,\ \omega '\hbox{ contains }y}
\end{equation*}
Due to the consistency of $\rho _n$, $D_n(x,y)$ is independent of $n$ if $n\ge\ab{x},\ab{y}$. Also $D_n(x,y)$, $x,y\in Q_n$, are the components of a positive semidefinite matrix. Moreover, if
\begin{equation*}
A_x=\brac{\omega\in\Omega _n\colon\omega\hbox{ contains }x}
\end{equation*}
then we define the $q$-measure $\mu _n(x)$ of $x$ by
\begin{equation*}
\mu _n(x)=\mu _n\sqbrac{\rmcyl (A_x)}=D_n(x,x)
\end{equation*}

We think of $Q_m$ as an analogue of a differentiable manifold and $D_n(x,y)$ as an analogue of a metric tensor.
For $\omega ,\omega '\in\Omega _n$ we define the \textit{covariant bidifference operator}
$\nabla _{\omega ,\omega '}^n\colon L_n\to L_n$ \cite{gud122} by
\begin{align*}
\nabla _{\omega ,\omega '}^nf(x,y)
  =&\sqbrac{D_n(\omega _{\ab{x}-1},\omega '_{\ab{y}-1})f(x,y)-D_n(x,y)f(\omega _{\ab{x}-1},\omega '_{\ab{y}-1})}\\
  &\quad\cdot\delta _{x,\omega _{\ab{x}}}\delta _{y,\omega '_{\ab{y}}}
\end{align*}
In analogy to the curvature operator on a manifold, we define the \textit{discrete curvature operator}
$\rscript _{\omega ,\omega '}^n\colon L_n\to L_n$ by
\begin{equation*}
\rscript _{\omega ,\omega '}^n=\nabla _{\omega ,\omega '}^n-\nabla _{\omega ',\omega}^n
\end{equation*}
We also define the \textit{discrete metric operator} $\dscript _{\omega ,\omega '}^n$ on $L_n$ by
\begin{align*}
\dscript&_{\omega ,\omega '}^nf(x,y)\\
  &=D_n(x,y)\sqbrac{f(\omega '_{\ab{x}-1},\omega _{\ab{y}-1})\delta _{x,\omega '_{\ab{x}}}\delta _{y,\omega _{\ab{y}}}
  -f(\omega _{\ab{x}-1},\omega '_{\ab{y}-1})\delta _{x,\omega _{\ab{x}}}\delta _{y,\omega '_{\ab{y}}}}
  \end{align*}
and the discrete mass-energy operator $\tscript _{\omega ,\omega '}^n$ on $L_n$ by
\begin{align*}
\tscript&_{\omega ,\omega '}^nf(x,y)\\
  &=\sqbrac{D_n(\omega _{\ab{x}-1},\omega '_{\ab{y}-1})\delta _{x,\omega _{\ab{x}}}\delta _{y,\omega '_{\ab{y}}}
  -D_n(\omega '_{\ab{x}-1},\omega _{\ab{y}-1})\delta _{x,\omega '_{\ab{x}}}\delta _{y,\omega _{\ab{y}}}}f(x,y)
  \end{align*}
It is not hard to show that

\begin{equation}         % equation (5.1)
\label{eq51}
\rscript _{\omega ,\omega '}^n=\dscript _{\omega ,\omega '}^n+\tscript _{\omega ,\omega '}^n
\end{equation}
We call \eqref{eq51} the \textit{discrete Einstein equation} \cite{gud122}

If we can find $D_n(\omega ,\omega ')$ such that the classical Einstein equation is an approximation to \eqref{eq51}, then it would give information about $D_n(\omega ,\omega ')$. Moreover, an important problem in discrete quantum gravity theory is how to test whether general relativity is a close approximation to the theory. Whether Einstein's equation is an approximation to \eqref{eq51} would provide such a test.

As with the classical Einstein equation \eqref{eq51} is difficult to analyze. We obtain a simplification by considering the
\textit{contractive discrete curvature}, \textit{metric} and \textit{mass-energy operators}
$\rhat _{\omega ,\omega '}^n,\dhat _{\omega ,\omega '}^n,\that _{\omega ,\omega '}^n\colon L_n\to K_n$, respectively, given by
$(\rhat _{\omega ,\omega '}^nf)(x)=\rscript _{\omega ,\omega '}^nf(x,x)$,
$(\dhat _{\omega ,\omega '}^nf)(x)=\dscript _{\omega ,\omega '}^nf(x,x)$,
$(\that _{\omega ,\omega '}^nf)(x)=\tscript _{\omega ,\omega '}^nf(x,x)$.
We then have the \textit{contracted discrete Einstein equation}
\begin{equation*}
\rhat _{\omega ,\omega '}^n=\dhat _{\omega ,\omega '}^n+\that _{\omega ,\omega '}^n
\end{equation*}
where
\begin{align*}
\dhat _{\omega ,\omega '}^nf(x)
  &=\mu _n(x)\sqbrac{f(\omega '_{\ab{x}-1},\omega _{\ab{x}-1})-f(\omega _{\ab{x}-1},\omega '_{\ab{x}-1}}\\
  &\quad\cdot\delta _{x,\omega '_{\ab{x}}}\delta _{x,\omega _{\ab{x}}}\\
  \that _{\omega ,\omega '}^nf(x)
  &=2i\rmim D_n(\omega _{\ab{x}-1},\omega '_{\ab{x}-1})\delta _{x,\omega _{\ab{x}}}\delta _{x,\omega '_{\ab{x}}}f(x,x)
\end{align*}

Any $f\in L_n$ can be decomposed into a sum of its symmetric and antisymmetric parts: $f=f_s+f_a$ where
\begin{align*}
f_s(x,y)&=\frac{f(x,y)+f(y,x)}{2}\\\noalign{\smallskip}
f_a(x,y)&=\frac{f(x,y)-f(y,x)}{2}
\end{align*}
and $f_s(x,y)=f_s(y,x)$, $f_a(x,y)=-f_a(y,x)$. We then obtain the simpler forms
\begin{align*}
\dhat _{\omega ,\omega '}^nf(x)
  &=\mu _n(x)f_a(\omega '_{\ab{x}-1},\omega _{\ab{x}-1})\delta _{x,\omega '_{\ab{x}}}\delta _{x,\omega _{\ab{x}}}\\
  \that _{\omega ,\omega '}^nf(x)
  &=2i\rmim D_n(\omega _{\ab{x}-1},\omega '_{\ab{x}-1})\delta _{x,\omega _{\ab{x}}}\delta _{x,\omega '_{\ab{x}}}f_s(x,x)
\end{align*}

\end{document}